\documentclass[a4paper,11pt,makeidx]{amsart} 
\usepackage{amscd}             
\usepackage{amssymb}
\usepackage{amsthm}
\usepackage{epsf} 
\usepackage{bm}
\usepackage[dvipsnames]{xcolor}
\usepackage[T1]{fontenc}   
\usepackage{tikz}

\usepackage[english]{babel}
\usepackage{url}
\usepackage{amssymb,latexsym,epsfig}
\usepackage[all]{xy}
\xyoption{all}
\usepackage{lmodern}
\usepackage[refpage]{nomencl}
\usepackage[dvipsnames]{xcolor}
\usepackage[utf8]{inputenc}
\usepackage{bm}
\usepackage{tikz}   
\DeclareMathOperator{\rk}{rk}        
\makeindex
\newtheorem{thm}{Theorem}[section]   
\newtheorem{cor}[thm]{Corollary}
\newtheorem{exa}[thm]{Example}
\newtheorem{lemma}[thm]{Lemma}
\newtheorem{prop}[thm]{Proposition}
\newtheorem{example}[thm]{Example}
\newtheorem{defn}[thm]{Definition}

\newtheorem{rem}[thm]{Remark}

\newcommand{\Fq}{\mathbb{F}_q}

\newcommand{\N}{\mathbb{N}}

\newcommand{\I}{\mathcal{I}}
\newcommand{\K}{\mathbb{K}}
\newcommand{\C}{\mathcal{C}}

\def\ker{\operatorname{ker}}

\def\min{\operatorname{min}}

\def\c1{\operatorname{c_1}}
\def\c2{\operatorname{c_2}}

\def\C{{\mathcal C}}

\def\N{{\mathcal N}}

\def\I{{\mathcal J}}

\def\K{{\mathcal K}}

\def\+{\oplus}                   
\def\*{\otimes}                  

\def\ker{\operatorname{Ker}}

\def\Fq{\mathbb{F}_{q}}

\begin{document}
\title[Greedy weights for matroids]{Greedy weights for matroids}

\author[Johnsen]{Trygve Johnsen}
\address{Department of Mathematics and Statistics, 
 UiT-The Arctic University of Norway  \newline \indent 
N-9037 Troms{\o}, Norway}
\email{trygve.johnsen@uit.no}
\thanks{}

\author[Verdure]{Hugues Verdure}
\address{Department of Mathematics and Statistics, 
 UiT-The Arctic University of Norway  \newline \indent 
N-9037 Troms{\o}, Norway}
\email{Hugues.Verdure@uit.no} 
\thanks{Both authors are  partially supported by grant 280731 from the Research Council of Norway. We are grateful to the  Department of Mathematics, IIT-Bombay for a stimulating stay, during which a substantial part of the present work was completed.}

\date{\today}

\begin{abstract}
We introduce greedy weights of matroids, inspired by those for linear codes. We show that a Wei duality holds for two of these types of greedy weights for matroids. Moreover we show that in the cases where the matroids involved are associated to linear codes, our definitions coincide with those for codes. Thus our Wei duality is a generalization of that for linear codes given by Schaathun. In the last part of the paper we show how 
some important chains of cycles of the matroids appearing, correspond to chains of component maps of minimal resolutions of the independence complex of the corresponding matroids. We also relate properties of these resolutions to chainedness and greedy weights of the matroids, and in many cases codes, that appear. 
\end{abstract}

\maketitle
\section{Introduction}

For a linear code over a finite field $\mathbb{F}_q$ an important way to characterize the code is to decribe its parameters, the word length $n$, the dimension $k$, and the minimum distance $d$. A refinement of the minimum distance is the ordered set of the generalized Hammimg weights $d_1,\cdots,d_k,$ where $d_i$ is the smallest support of any $i$-dimemsional linear subcode of $C$, for $i=1,\cdots,k$. In particular $d_1=d$.
In the 1990's (and early 2000's) several authors (see e.g. \cite{chen1999weight}, \cite{Cohen1999} \cite{Chen04}, \cite{Chen01}, \cite{Schaathun01}. \cite {Schaathun04}, \cite{Schaathun01phd}) became interested not only
in the individual subcodes of each dimension that where optimal with respect to (small) support size, but also in  chains of codes that where somehow optimal, in a similar way. This gave rise to various definitions of greedy weights, which we will recall in Subsection \ref{basicdef}. These weights are similar to, but in general different from, the generalized Hamming weights $d_i$. The topic has attracted new interest in recent years (\cite{Li17}, \cite{Bai19}).

In \cite{Johnsen13} we described how the $d_i$ are determined by certain properties of the matroid coming from any parity check matrix of the linear code. In the present paper we will describe how also the various greedy weights are determined by the matroids associated to the code. Since this description can be done for any finite matroid , we will define and describe greedy weights for finite matroids in general,
and show that they coincide with those of linear codes when the matroids come from such codes.
We will show a form of Wei duality relating certain weights of a matroid and its dual, inspired by a corresponding result for 
linear codes (\cite{Schaathun01phd}). 
 
The spirit of the paper is the following: There is a poset of cycles of the matroid coming from any parity check matrix of the code, where a cycle is an inclusion minimal set among those subsets of $E=\{1,\cdots,n\}$ having a fixed nullity for the rank function in question. This is dual to (the upside down version of) the poset of flats  of the matroid coming from any generator matrix of the code. We will show that the greedy weights correspond to optimal ways to traverse the nodes of this poset through maximal chains of it.
We define a lexicographical and a rev-lexicographical order on these chains in order to make it precise 
in what sense they are optimal. 

In the last part we relate our results to a more concrete way to traverse maximal chains via non-zero component maps in a minimal resolution of a certain Stanley-Reisner ring, where the components in each fixed step corresponds to the nodes of a corresponding fixed rank of the poset of cycles.

 This paper is organized as follows. In Section \ref{defs} we will give some necessary definitions relating to codes and matroids. In Section \ref{important} we will
describe the greedy weights for matroids, relate to those of codes, and show our form of Wei duality, which is inspired by the corresponding Wei duality for codes, proven in \cite{Schaathun01phd}.
In Section \ref{main} we discuss the connection between resolutions of the Stanley-Reisner ring associated to the matroid or the code, and the greedy weights. We also discuss the notion of chained codes and chained matroids.

The main results are Theorems \ref{firstmain}, \ref{th:coincide} and \ref{nonzerocomp}.

%
%

\section{Definitions and notation} \label{defs}

\subsection{Generalized Hamming weights and greedy weights of codes}\label{basicdef}

\begin{defn}\label{def:support}
Let $C$ be a $[n,k]$ linear code over $\Fq$. Let $\bm{c} = (c_1,\cdots,c_n) \in C$. The Support of $\bm{c}$ is the set \[Supp(\bm{c}) = \{ i \in \{1,\cdots,n\}: c_i \neq 0\}.\] Its weight is \[wt(\bm{c}) = |Supp(\bm{c})|.\] Similarly, if $T \subset C$, then its support and weight are \[Supp(T) = \bigcup_{\bm{c} \in T} Supp(\bm{c}) \textrm{ and } wt(T) = |Supp(T)|.\]
\end{defn}

In \cite{Chen1999} and \cite{Bai19} one describes and treats greedy weights of linear codes $C$ over finite fields.
First we recall the definitions of the generalized Hamming weights introduced by Wei in~\cite{Wei1991}:
\begin{defn}\label{def:Wei}
Let $C$ be a $[n,k]$-linear code. For $1 \leqslant r \leqslant k$, the $r$-th generalized Hamming weight is \[d_{r}= \min \{w(D) | D\textrm{ is a subcode of }C\textrm{ with }\dim D=r\}.\]
\end{defn}

A subcode $D \subset C$ \emph{computes} $d_r$ if it is of dimension $r$ and weight $d_r$.

Then, following the terminology of~\cite{Schaathun01} or~\cite{Schaathun04}, we have the (bottom-up) greedy weights of a code:

\begin{defn}\label{def:bottomupsubcode}
Let $C$ be a $[n,k]$-linear code. A (bottom-up) greedy $1$-subcode is a subcode of dimension $1$ of minimal weight. For $r\geqslant 2$, a (bottom-up) greedy $r$-subcode is a subcode of dimension $r$ containing a (bottom-up) greedy $(r-1)$-subcode, and such that no other such subcode has lower weight.
\end{defn}

\begin{defn}\label{def:bottomup}
Let $C$ be a $[n,k]$-linear code. For $1 \leqslant r \leqslant k$, the  $r$-th (bottom-up) greedy weight $e_r$ of $C$ is the weight of any (bottom-up) greedy $r$-subcode.
\end{defn}

\begin{rem}
We have $e_1=d_1$.
\end{rem}

Also introduced by Schaathun (\cite[Definition 6]{Schaathun01}) are the top down greedy weights:

\begin{defn}\label{def:topdownsubcode}
Let $C$ be a $[n,k]$-linear code. A top-down greedy $k$-subcode is $C$ itself. For $r\leqslant k-1$, a top-down greedy $r$-subcode is a subcode of dimension $r$ contained in a top-down greedy $(r+1)$-subcode, and such that no other such subcode has lower weight.
\end{defn}

\begin{defn}\label{def:topdown}
Let $C$ be a $[n,k]$-linear code. For $1 \leqslant r \leqslant k$, the  $r$-th top-down greedy weight $\tilde{e}_r$ of $C$ is the weight of any top-down greedy $r$-subcode.
\end{defn}

\begin{rem}
We have $\tilde{e}_k=d_k$.
\end{rem}

There is also another definition, used e.g by \cite{Bai19}, essentially introduced in \cite{Cohen1999}:

\begin{defn}\label{def:CEZsubcode}Let $C$ be a $[n,k]$-linear code.  A CEZ greedy $1$-subcode is a subcode of dimension $1$ of minimal weight. For $r\geqslant 2$, a CEZ greedy $r$-subcode is a subcode of dimension $r$ containing a subcode that computes $d_{r-1}$, and such that no other such subcode has lower weight.
\end{defn}

\begin{defn}\label{def:CEZ} Let $C$ be a $[n,k]$-linear code. For $1 \leqslant r \leqslant k$, the  $r$-th CEZ greedy weight $g_r$ of $C$ is the weight of any CEZ greedy $r$-subcode.
\end{defn}
\begin{rem} We have $g_1=e_1=d_1$ and $g_2=e_2$
\end{rem}

For more interesting material on this topic, see \cite{Chen04}, \cite{Chen01}, \cite{chen1999weight}, \cite{Li17}, \cite{Schaathun01}.

\subsection{Matroids}

There are many equivalent definitions of a matroid. We refer to~\cite{Oxley11} for a deeper study of the theory of matroids. 

\begin{defn}\label{def:matroid}
A matroid is a pair $M=(E,r)$ where $E$ is a finite set and $r: 2^E \rightarrow \N$ is a function, called the rank function, satisfying \begin{itemize}
\item[($R_1$)] If $X \subset E$, then \[0 \leqslant r(X) \leqslant |X|,\]
\item[($R_2$)] If $X \subset Y \subset E$ then \[r(X) \leqslant r(Y),\]
\item[($R_3$)] If $X,Y$ are subsets of $E$, then \[r(X \cap Y) + r(X \cup Y) \leqslant r(X) + r(Y).\] The rank of the matroid is $r(M)=r(E)$.
\end{itemize}

\end{defn}

It is a well known fact the rank  function of a matroid is \emph{unit rank increase}, that is, if $X \subset E$ and $x \in E$, then \[r(X) \leqslant r(X \cup \{x\}) \leqslant r(X)+1.\]

\begin{defn}\label{def:nullity}
The nullity function of the matroid $(E,r)$ is the function defined on $2^E$ by: for $X \subset E$, \[n(X)=|X| - r(X).\]
\end{defn}

The nullity function of a matroid is also unit rank increase. Moreover, it satisfies $(R_1)$, $(R_2)$ as well as \[n(X \cap Y) + n(X \cup Y) \geqslant n(X) + n(Y)\] for $X,
Y$ subsets of $E$.

\begin{defn}\label{def:dual}
Let $M=(E,r)$ be a matroid. Then its dual matroid is the matroid $\overline{M}=(E,\overline{r})$ where $\overline{r}$ is defined by \[\overline{r}(X) = |X| + r(E\-X) - r(E)\] for $X \subset E$.
\end{defn}

Some subsets of the ground set of a matroid will be of special interest in this article:

\begin{defn}\label{def:circuits}
Let $(E,r)$ be a matroid. A subset $X \subset E$ is dependent if \[r(X)<|X| \Leftrightarrow n(X)>0\] and independent if \[r(X)=|X| \Leftrightarrow n(x)=0.\] A \emph{circuit} is a inclusion minimal dependent set. We denote by $\I$ and $\C$ the sets of independent sets and circuits respectively. 
\end{defn}

For $1 \leqslant i \leqslant |E|-r(M)$ will denote by $\mathcal{N}_i$ the set \[\mathcal{N}_i = \{X \subset E,\ n(X)=i\}\] and by $N_i$ the inclusion minimal elements of $\mathcal{N}_i$. It is clear that \[\mathcal{C} = N_1.\]
A \emph{cycle} is an element of $N_i$ for some $i$. Cycles can also be described as unions of circuits, and the nullity of the cycle is equal to the maximal number of non-redundant circuits in the cycle (\cite{Johnsen13}).

If $C$ is a $[n,k]$-linear code given by a $(n-k) \times k$ parity check matrix $H$, then we can associate to it a matroid $M_C=(E,r)$, where $E=\{1,\cdots ,n\}$ and if $X \subset E$, then \[r(X)=\rk(H_X),\] where $H_X$ is the column submatrix of $H$ indexed by $X$. It can be shown that this matroid is independent of the choice of the parity check matrix of the code, and we may thus call it the matroid of $C$.

\subsection{Resolutions}

If $(M,r)$ is a matroid, then $(E,\I)$ is naturally a simplicial complex (that is, $\I \neq \emptyset$ and is closed under taking subsets). Let $\K$ be a field. We can associate to $M$  a monomial ideal $I_M$ in $S=\K[\{X_e\}_{e \in E}]$ defined by \[I_M =< \bm{X}^\sigma: \sigma \in \C>\] where $\bm{X}^\sigma$ is the monomial product of all $X_e$ for $e \in \sigma$. This ideal is called the Stanley-Reisner ideal of $M$ and the quotient $S_M=S/I_M$ the Stanley-Reisner ring associated to $M$. We refer to~\cite{Herzog11} for the study of such objects. As described in \cite{Johnsen13} the Stanley-Reisner ring has minimal $\mathbb{N}$ and $\mathbb{N}^n$-graded free resolutions 

\[
0 \leftarrow S_M \leftarrow S \leftarrow \bigoplus_{j \in \N}S(-j)^{\beta_{1,j}} \leftarrow 
\cdots \leftarrow \bigoplus_{j \in \N}S(-j)^{\beta_{|E|-r(M),j}} \leftarrow 0
\]

and \[0 \leftarrow S_M \leftarrow S \leftarrow \bigoplus_{\alpha \in \N^n}S(-\alpha)^{\beta_{1,\alpha}} 
\leftarrow \cdots \leftarrow \bigoplus_{\alpha \in \N^n}S(-\alpha)^{\beta_{|E|-r(M),\alpha}} \leftarrow 0.\] 

 In particular the numbers $\beta_{i,j}$ and $\beta_{i,\alpha}$ are independent of the minimal free resolution,
 (and for a matroid also of the field $\K$) and are called respectively the $\mathbb{N}$-graded and $\mathbb{N}^n$-graded Betti numbers of the matroid. Note also that if $\alpha \not \in \{0,1\}^n$, then $\beta_{i,\sigma}=0$ (\cite[Corollary 1.40]{miller04}).
 We have \[\beta_{i,j} = \sum_{wt(\alpha)=j}\beta_{i,\alpha}.\] We also note that $\beta_{0,0}=1$.

We will also frequently use (\cite[Theorem 1]{Johnsen13}, first part): 
\begin{thm} \label{thm:AAECC}
Let $C$ be a $[n,k]$-code over $\Fq$. The $\mathbb{N}$-graded Betti numbers of the matroid $M_C$ satisfy:  $\beta_{i,j} \neq 0$ if and only if there exists a member in $N_i$ of cardinality $j$. In particular, $\beta_{i,X} \ne 0$ if and only $X \in N_i$. Furthermore \[d_i = \min\{j: \beta_{i,j} \neq 0\}.\]
\end{thm}
\begin{rem}
The fact that $\beta_{i,X} \ne 0$ if and only $X \in N_i$ is a consequence of the considerations on \cite[page 59]{stanley77}, where one also relates these Betti numbers to M\"obius numbers of related lattices of cycles.
\end{rem}
%
%

\section{Greedy weights for matroids} \label{important}

We will now give definitions for greedy weights for matroids, and later show that greedy weights for linear codes and their associated matroids coincide. First, recall the definition for generalized Hamming weights for matroids, given in~\cite{Larsen05}:
\begin{defn}\label{def:hammingM} Let $M$ be a matroid of rank $n-k$ on a set of cardinality $n$. For $1 \leqslant r \leqslant n-k,$
\[d_r= \min\{|\sigma|:\ \sigma \in \mathcal{N}_r\} = \min\{|\sigma|:\ \sigma \in N_r\}.\]
\end{defn}

\begin{defn}\label{def:greedyM}
Let $M$ be a matroid on $n$ elements of rank $n-k$. Let $\Sigma$ be the set \[\Sigma =\left\{ (\sigma_1,\cdots,\sigma_k) \in \mathcal{N}_1\times\cdots\times \mathcal{N}_k\vert\ \sigma_1 \subsetneq \cdots \subsetneq \sigma_k\right\}.\] Let $\overline{\Sigma}$ be the set \[\overline{\Sigma} = \left\{e(S)=(|\sigma_1|,\cdots,|\sigma_k|):\ S=(\sigma_1,\cdots,\sigma_k) \in \Sigma\right\}.\] Then the (bottom-up) greedy weights $(e_1,\cdots,e_k)$ of $M$ are the \[(e_1,\cdots,e_k) = \min_{lex} \overline{\Sigma}\] while the top-down greedy weights $(\tilde{e}_1,\cdots,\tilde{e}_k)$ of $M$ are \[(\tilde{e}_1,\cdots,\tilde{e}_k) = \min_{revlex} \overline{\Sigma},\] where \emph{lex} and \emph{revlex} are the lexicographic and reverse lexicographic orders respectively.
\end{defn}

If $S=(\sigma_1,\cdots,\sigma_k) \in \mathcal{N}_1\times\cdots\times\mathcal{N}_k$ is such that $e(S)=(e_1,\cdots,e_k)$ (resp. $(\tilde{e}_1,\cdots,\tilde{e}_k)$), we say that $\sigma_i$ \emph{computes} $e_i$ (resp. $\tilde{e}_i$).

\begin{defn}\label{def:CEZM} Let $M$ be a matroid of rank $n-k$ on a set of cardinality $n$, and let $(d_1,\cdots,d_k)$ be its generalized Hamming weights. The CEZ greedy weights $(g_1,\cdots,g_k)$ are defined as follows: \[g_1=d_1\] and for $2\leqslant r \leqslant k$,
 \[g_r = \min \{|\sigma|:\ \sigma \in \mathcal{N}_r \textrm{ and } \exists \tau \in \mathcal{N}_{r-1} \textrm{ such that } \tau \subset \sigma \textrm{ and }|\tau| = d_{r-1}\}.\] 
\end{defn}

We say that $\sigma \in \mathcal{N}_i$ \emph{computes} $g_i$ if it satisfies the conditions in the definition.

\begin{exa}\label{exa:code}
Let $C$ be the $[8,4]$-linear code over $\mathbb{F}_3$ defined by the generator matrix \[G=\begin{bmatrix}
1&0&1&1&0&0&0&0\\
0&1&1&1&0&0&0&0\\
0&0&0&0&1&1&1&0\\
0&0&0&0&1&2&0&1
\end{bmatrix}.\] Its weights are
\[(d_1,d_2,d_3,d_4)=(2,4,6,8),\] \[(e_1,e_2,e_3,e_4)=(g_1,g_2,g_3,g_4) = (2,4,7,8)\] and \[(\tilde{e}_1,\tilde{e}_2,\tilde{e}_3,\tilde{e}_4,)=(3,4,6,8).\]
\end{exa}

As a consequence of the unique rank increase of the nullity function, both the bottom up and the top down greedy weights are strictly increasing sequences. The CEZ greedy weights $g_i$ are not necessary monotonous, as the following example shows.

\begin{example} \label{nonincreasing}
Let $M$ on $E=\{1,\cdots,23\}$ whose circuits are the following: all the subsets of $\{13,\cdots,23\}$ of cardinality $9$ together with $\{1,\cdots,8\}$, $\{5,\cdots,12\}$ and $\{1,2,3,4,9,10,11,12\}$. This is a matroid of rank $18$. Then, \[(d_1,d_2,d_3,d_4,d_5) = (8,10,11,19,23),\]\[(e_1,e_2,e_3,e_4,e_5) = (8,12,21,12,23),\]\[(\tilde{e}_1,\tilde{e}_2,\tilde{e}_3,\tilde{e}_4,\tilde{e}_5) = (10,11,12,19,23),\]\[(g_1,g_2,g_3,g_4,g_5) = (8,12,11,19,23).\]
\end{example}

In Definitions~\ref{def:greedyM} and ~\ref{def:CEZM}, we could actually have asked the subsets to be in $N_i$, not just $\mathcal{N}_i$, as the following proposition shows:

\begin{prop}\label{prop:2defs}
Let $M$ be a matroid of rank $n-k$ on a set of cardinality $n$. Let $\Sigma'$ be the set \[\Sigma' =\left\{ (\sigma_1,\cdots,\sigma_k):\ \sigma_1 \subsetneq \cdots \subsetneq \sigma_k \textrm{ and }\sigma_i \in N_i\textrm{, }\ \forall i\right\}.\]Then we have the following:
\[(e_1,\cdots,e_k)=\min_{lex}\{e(S):\ S \in \Sigma'\},\]
\[(\tilde{e}_1,\cdots,\tilde{e}_k)=\min_{revlex}\{e(S):\ S \in \Sigma'\},\] and for all $2\leqslant i \leqslant k$,
 \[g_i=\min\{|\sigma|:\ \sigma \in N_i \textrm{ and } \exists \tau \in N_{i-1} \textrm{ such that } \tau \subset \sigma \textrm{ and }|\tau| = d_{i-1}\}.\] 
\end{prop}

\begin{proof}
The first and third assertions rely on the same observation. We will thus only treat the first assertion. It is clear that \[\min_{lex}\{e(S):\ S \in \Sigma'\} \geqslant_{lex} (e_1,\cdots,e_k).\] Now, let $S=(\sigma_1,\cdots,\sigma_k) \in \Sigma$ such that \[e(S)=\min_{lex}\{e(S):\ S \in \Sigma'\}.\] We will show that $\sigma_i \in N_i$ for all $i$. If not, let $i$ be the smallest index for which this is not true. By Definition~\ref{def:hammingM}, $i >1$. Since $\sigma_i \not\in N_i$, this means that there exists $\tau \subsetneq \sigma_i$ such that $n(\tau)=i$. Obviously, $\sigma_{i-1} \not\subset \tau$ otherwise, replacing $\sigma_i$ by $\tau$ in the sequence $S$, we would get a chain of sets that would contradict the minimality of $e(S)$ for the \emph{lex} ordering. Thus, we can find $x \in \sigma_{i-1} - \tau$. Without loss of generality, we can suppose that $\tau=\sigma_i-\{x\}$. Consider then $\rho=\sigma_{i-1}-\{x\}$. By minimality of $\sigma_{i-1}$ in the set of subsets with nullity $i-1$, and by the unique rank increase property of $n$, $n(\rho) = i-2$. Then we have, by the inequality after Definition~\ref{def:nullity} satisfied by the nullity function: \[2i-2= n(\rho ) + n(\sigma_i) = n(\sigma_{i-1}\cap \tau) + n(\sigma_{i-1} \cup \tau) \geqslant n(\sigma_{i-1}) + n(\tau) = 2i-1,\] which is absurd.
Thus, all elements in $S$ are in $N_i$, and the first assertion is proved.\\
The second assertion is easier to prove since we don't have any bottom constraints. Again, it is clear that \[\min_{revlex}\{e(S):\ S \in \Sigma'\} \geqslant_{lex} (\tilde{e}_1,\cdots,\tilde{e}_k).\] For the contrary, let  $S=(\sigma_1,\cdots,\sigma_k) \in \Sigma'$ such that \[e(S)=\min_{revlex}\{e(S):\ S \in \Sigma\}.\] Assume that there exists an index $i$ such that $\sigma_i \not \in  N_i$. Let $\tau_i \subsetneq \sigma_i$ such that $\tau_i \in N_i$, and take recursively for $j<i$ any $\tau_j \subset \tau_{j+1}$ such that $n(\tau_j)=j$. 
This can always be done by the unique rank increase property of $n$. Then the sequence $S'=(\tau_1,\cdots,\tau_i,\sigma_{i+1},\cdots,\sigma_k) \in \Sigma$, and by construction, \[e(S') <_{revlex} e(S),\] which is absurd. 
This in turn shows that \[\min_{revlex}\{e(S):\ S \in \Sigma'\} \leqslant_{lex} (\tilde{e}_1,\cdots,\tilde{e}_k).\]
\qed \end{proof}

\begin{rem}
The set $\Sigma$ appearing in Proposition \ref{prop:2defs} is the set of maximal chains in the poset of cycles for the matroid. Taking complements, this is the poset of flats of the dual matroid. If $d^{\perp} \ge 3,$ then this poset is a geometric lattice with atoms of cardinality $1$. Then the cardinalities $c_f$ of the flats, and hence all the cardinalities $n-c_f$ of the cycles $\sigma$ of the matroid, can be given a purely lattice-theoretical interpretation in terms of atoms.  Hence it is possible to reformulate Proposition \ref{prop:2defs} by lattice-theoretical invariants.
\end{rem}

\begin{cor}\label{cor:greedyamongbetti}
Let $M$ be a matroid of rank $n-k$ on a set of cardinality $n$. For $1 \leqslant i \leqslant k$, \[X \subset E \textrm{ is a (top-down, bottom-up, CEZ) }i\textrm{-greedy subcode} \Rightarrow \beta_{i,X} \neq 0\] and \[g_i, e_i,\tilde{e}_i \in \{j \vert \beta_{i,j} \neq 0\}.\]
\end{cor}

\begin{proof}
In the proof above, we showed that any subset that computes a greedy-weight is a cycle. This is then a direct consequence of Theorem~\ref{thm:AAECC}.
\qed \end{proof}

\subsection{Wei duality of greedy weights}

If $M$ is a matroid, then it is proved in~\cite{Larsen05} that the weight hierarchy of the matroid and its dual satisfy Wei duality, that is \[\{d_1,\cdots,d_k\} \cup \{n+1-\overline{d}_1,\cdots,n+1-\overline{d}_{n-k}\} = \{1,\cdots,n\},\] 
where $\overline{d}_i$ denotes the $i$-th generalized Hamming weight of $\overline{M}$. This result is a generalization of duality for linear codes proved by Wei (\cite{Wei1991}). In his doctoral thesis (\cite{Schaathun01phd}), Schaathun proves a Wei duality for greedy weights for linear codes, namely that \[\{e_1,\cdots,e_k\} \cup \{n+1-\overline{\tilde{e}}_1,\cdots,n+1-\overline{\tilde{e}}_{n-k}\} = \{1,\cdots,n\}.\] In this section, we will prove that his result extends to matroids. Before doing so, if $S=(\sigma_1,\cdots,\sigma_k) \in \Sigma$, we define $\delta(S)$ in the following (not unique) way: consider a maximal chain \[\emptyset \subsetneq \rho_1 \subsetneq \cdots \subsetneq \rho_n=E\] that contains all the $E-\sigma_i$ for $1\leqslant i\leqslant k$. Obviously, we have $|\rho_i|=i$ for every $1\leqslant i\leqslant n$. Then $\delta(S)$ is the chain $\tau_1\subsetneq\cdots\subsetneq\tau_{n-k}$ obtained by removing all the subsets of cardinality $n-|\sigma_i|+1$. Even if this is not uniquely defined, the set $\overline{\delta}(S) = \{|\tau_1|,\cdots,|\tau_{n-k}|\}$ is, since we have \[\overline{\delta}(S) = E-\{n+1-|\sigma_i|: \ 1\leqslant i \leqslant k\}.\] In particular, we have, with a slight abuse of notation, \[\overline{\delta}\delta S = e(S)=\{|\sigma_1|,\cdots,|\sigma_k|\}.\] 
Denote by $\overline{n}$ the nullity function of $\overline{M}$.
\begin{lemma}\label{lem:Wei1}
Let $S=(\sigma_1,\cdots,\sigma_k)$ be a tower that computes the bottom up greedy weights of $M$, and let $\delta(S)=(\tau_1\cdots,\tau_{n-k})$. Then for all $1\leqslant i \leqslant n-k$, \[\overline{n}(\tau_i)=i.\]
\end{lemma}
\begin{proof}
Using the notation from the definition of $\delta(S)$, we have for every $i$ the chain \[E- \sigma_{i+1}= \rho_j \subsetneq \cdots \subsetneq \rho_{j+s} = E-\sigma_i\] where $j= n-|\sigma_{i+1}|$ and $s=|\sigma_{i+1}|-|\sigma_i|$. From the duality formula for the rank functions and nullity functions, we get that, since $n(\sigma_t)=t$, \[\overline{n}(E-\sigma_{i+1}) = k+i+1-|\sigma_{i+1}|\] while \[\overline{n}(E-\sigma_i)=k+i-|\sigma_i|.\] Since $\overline{n}$ is unit rank increase, this means that all $\overline{n}(\rho_{j+t})$ are distinct, except for $2$ of them, and that they span the set $\{k+i+1-|\sigma_{i+1}|,\cdots,k+i-|\sigma_i|\}.$ We show now that $\overline{n}(\rho_j)=\overline{n}(\rho_{j+1})$. Since both set differ by just $1$ element, we have either $\overline{n}(\rho_j)=\overline{n}(\rho_{j+1})$ or $\overline{n}(\rho_j)=\overline{n}(\rho_{j+1})-1$. Suppose the latter occurs. Then, \[n(\sigma_{i+1})=n(E-\rho_j) = n-k-|\rho_j|+\overline{n}(\rho_j) = n(E-\rho_{j+1}).\] Since \[\sigma_i \subsetneq E-\rho_{j+1} \subsetneq E-\rho_j = \sigma_{i+1}\] (the first strict inclusion coming from the fact that $n(\sigma_i)=n(\sigma_{i+1})-1 = n(E-\rho_{j+1})-1$), the tower \[\sigma_1\subsetneq \cdots\subsetneq \sigma_i \subsetneq E-\rho_{j+1} \subsetneq \sigma_{i+2} \cdots\subsetneq \sigma_k \in \Sigma\] and the $k$-tuple formed by the cardinalities of the elements of the tower is strictly lower for the \emph{lex} order than $(e_1,\cdots,e_k)$ which is absurd.
\qed \end{proof}

\begin{lemma}\label{lem:Wei2} Let $S,S' \in \Sigma$. Then \[e(S) <_{lex} e(S') \Leftrightarrow e(\delta(S))<_{revlex} e(\delta(S')).\]
\end{lemma}
\begin{proof}
Write $S=(\sigma_1,\cdots,\sigma_k)$, $S'=(\sigma'_1,\cdots,\sigma'_k)$, $\delta(S)=(\tau_1,\cdots,\tau_{n-k})$ and $\delta(S')=(\tau'_1,\cdots,\tau'_{n-k})$. By hypothesis, there exists an $1\leqslant i \leqslant k$ such that for all $1\leqslant j <i$, $|\sigma_j|=|\sigma'_j|$ while $|\sigma_i|<|\sigma'_i|$. In our definition of $\delta$ above (and we keep the notation, using $\rho_s$ and $\rho'_s$ for $S$ and $S'$ respectively) this means that for $l\geqslant n-|\sigma_i|-k+i+1$, \[|\tau_l|=|\tau'_l|\] while \[\left|\tau_{n-|\sigma_i|-k+i}\right|<\left|\tau'_{n-|\sigma_i|-k+i}\right|=n-|\sigma_i|+1,\] that is \[e(\delta(S))<_{revlex} e(\delta(S'))\]
 The other way is done in a similar way, noticing that $e\delta\delta(S)=\overline{\delta}\delta(S) = e(S)$.
\qed \end{proof}
We then obtain the following analogue of \cite[Theorem 10.2]{Schaathun01phd}, where one showed Wei duality for greedy weights of linear codes:
\begin{thm} \label{firstmain}
Let $M$ be a matroid of rank $k$ on a  ground set $E$ of cardinality $n$. Then \[\{e_1,\cdots,e_k\} \cup \{n+1-\overline{\tilde{e}}_1,\cdots,n+1-\overline{\tilde{e}}_{n-k}\} = \{1,\cdots,n\}.\]
\end{thm}

\begin{proof}
Let $S\in \Sigma$ such that $e(S)=(e_1,\cdots,e_k)$. Consider $T=\delta(S)$. By Lemma~\ref{lem:Wei1}, we know that $T \in \Sigma(\overline{M})$, and thus \[e(T) \geqslant_{revlex} (\overline{\tilde{e}}_1,\cdots,\overline{\tilde{e}}_{n-k}).\] If this is not an equality, let $T' \in \Sigma(\overline{M})$ such that $e(T')=(\overline{\tilde{e}}_1,\cdots,\overline{\tilde{e}}_{n-k})$. Then by Lemma~\ref{lem:Wei2} and the fact that $e\delta\delta (T)=e(T)$, we get that \[e(S) >_{lex} e(\delta(T)) \geqslant_{lex} (e_1,\cdots,e_k)=e(S)\] which is absurd.
\qed \end{proof}

\subsection{Greedy weights of codes and matroids} 

In for example  \cite{chen1999weight},  \cite{Chen1999} \cite{Cohen1999}  \cite{Chen01}, \cite{Schaathun01}, \cite {Schaathun04}, \cite{Schaathun01phd},  \cite{Chen1999}, \cite{Chen04} and  \cite{Bai19} one describes and treats greedy weights of linear codes $C$ over finite fields in various ways. In this part, we will show that the greedy weights for codes and their associated matroids coincide. We start with some lemmas:

\begin{lemma}\label{lem:techn}
Let $C$ be a $[n,k]$-code, $M$ its associated matroid and $X \subset \{1,\cdots,n\}$. Consider the subcode  \[C(X) = \left\{w \in C\vert\ Supp(w) \subset X\right\} \subset C.\] Then \[\dim C(X)=n(X)=n(Supp(C(X))).\]  Moreover, \[Supp (C(X)) = X \Leftrightarrow X \in N_{n(X)}.\] 
\end{lemma}

\begin{proof}
The first assertion is an easy consequence from the fact that $C(X) = \ker G_{E-X}$, and a rewriting of the rank-nullity theorem using the relation between the rank of the matroid and its dual. \\
From the previous assertion, the dimension of the relations between the columns of $H$ indexed $Supp(C(X))$ is $n(X)$, that is, \[n(Supp (C(X))) = n(X).\]
Finally, let $i=n(X)$. We have always $SuppC(X)) \subset X$. If $Supp C(X) \subsetneq X$, from what we have just seen, $n(Supp C(X))=i$, so that $X \not \in N_i$. Conversely, suppose that $X \not \in N_i$. Let $Y \subsetneq X$ in $N_i$, and consider the two subcodes $C(Y)$ and $C(X)$. Obviously, $C(Y) \subset C(X)$. By the first result of this lemma, they have the same dimension, so they have to be equal. Moreover, since $Y \in N_i$, $Supp (C(Y))=Y$. This shows that \[Supp(C(X)) = Supp C(Y) = Y \subsetneq X.\]
\qed \end{proof}

\begin{lemma}\label{lem:coden} Let $C$ be a $[n,k]$ linear code over $\Fq$. Let $1\leqslant i\leqslant k$ and let $D \subset C$ be a $i$-greedy subcode. Then \[Supp(D) \in \mathcal{N}_i.\]
\end{lemma}

\begin{proof} First of all, a codeword is a dependence relation between the columns of $H$, and by definition of the support, this is actually a dependence relation between the submatrix of $H$ indexed by the support. Saying that the subcode $D$ has dimension $i$ implies that there are at least $i$ independent relations between these columns, that is $n(Supp(D)) \geqslant i$. We will now prove that there is equality in the three cases of greedy subcodes.\\
When $D$ is a greedy $1$-subcode, a CEZ greedy $1$-subcode, a top-down greedy $i$-subcode, or a subcode that compute $d_i$, suppose that $n(Supp(D))\geqslant i+1$. Then there exists $X \subsetneq Supp(D)$ such that $n(X)=i$. Consider the subcode $C(X)  \subsetneq D$. Since it is a strictly smaller subcode than $D$, it has dimension at most $i-1$. At the same time, since $n(X)=i$, there are $i$ independent relations between the columns indexed by $X$, that is, the dimension of $C(X)$ is $i$, which is absurd.\\
Now, suppose that $i>1$ andlet $D$ be a bottom-up or a CEZ $i$-greedy subcode. Then there exists a subcode $D' \subset D$ which is either a greedy $(i-1)$-subcode or that computes $d_{i-1}$. In any case, by what we have just proved, $n(Supp(D'))=i-1$. If $n(Supp(D))>i$, then there exists a set $X$ such that $n(X)=i$ and \[Supp(D') \subsetneq X \subsetneq Supp(D).\] Then \[D' \subset C(X) \subsetneq D.\] Thus, \[i \leqslant \dim C(X) < \dim D = i\] which is absurd.

\qed \end{proof}

For related results, see \cite[Section 3]{GhorpadeSingh20}.
We have actually a stronger result, namely:

\begin{cor} \label{l1}
Let $C$ be a $[n,k]$ linear code over $\Fq$. Let $1\leqslant i\leqslant k$ and let $D \subset C$ be a $i$-greedy subcode. Then \[Supp(D) \in N_i.\]
\end{cor}

\begin{proof}
This a consequence a Lemma~\ref{lem:coden} and the same procedure we did in the proof of the Proposition~\ref{prop:2defs}. We look at the lowest $i$ such that $Supp(D_i)$ is not in $N_i$. Consider the two subcodes $D_{i-1} \subset D_i$. Then $X=Supp (D_{i-1})$ is in $N_{i-1}$, while $Y=Supp(D_i)$ has nullity $i$, but is not in $N_i$. Thus, there exists another subset $Z$ such that $Z \subsetneq Y$ in $N_i$. Of course $X \not\subset Z$, otherwise $C(Z)$ will contradict the minimality of $|Supp(D_i)|$ by Lemma~\ref{lem:techn}. Take $Z' = Y-\{x\}$ for a $x \in X-Z$. Then $n(Z')=i$. By minimality of $X$, we thus have $n(X \cap Z') = i-2$. Then we have \[2i-2 \geqslant n(Y)+n(X \cap Z') \geqslant n(X) + n(Z') = 2i-1\] which is absurd.
\qed \end{proof}

\begin{thm}\label{th:coincide}
The greedy weights of a $[n,k]$-linear code $C$ and its associated matroid coincide.
\end{thm}

\begin{proof}
From Lemma~\ref{lem:coden} and Definitions~\ref{def:greedyM} and~\ref{def:CEZM}, we have
 \begin{eqnarray*}(e_1(C),\cdots,e_k(C)) &\geqslant_{lex}& (e_1(M),\cdots,e_k(M)),\\
(\tilde{e}_1(C),\cdots,\tilde{e}_k(C)) &\geqslant_{revlex}& (\tilde{e}_1(M),\cdots,\tilde{e}_k(M))\end{eqnarray*} and for every $1\leqslant i \leqslant k$, \[g_i(C) \geqslant g_i(M).\] 
Let $S=(\sigma_1,\cdots,\sigma_k)\in \Sigma'$ be such that $e(S) = (e_1(M),\cdots,e_k(M))$, and consider the  subcodes $D_i=C(\sigma_i)$. From Lemma~\ref{lem:techn}, we know that $Supp(D_i)=\sigma_i \in N_i$ and $\dim D_i = n(\sigma_i)=i$.
Clearly the $D_i$ is a chain of linear codes totally ordered by inclusion, with
$(|Supp(D_1)|,\cdots,|Supp(D_k)|)=(e_1(M),\cdots,e_k(M),$ so \[(e_1(C),\cdots,e_k(C)) \leqslant_{lex} (e_1(M),\cdots,e_k(M)).\]
The proofs for top-down and CEZ greedy weights are done in a  similar way.

\qed \end{proof}

%
%

\section{Greedy weights and resolutions of Stanley-Reisner rings} \label{main}

Let $M$ be a matroid of rank $n-k$ over a finite set $E$ of cardinality $n$ (for example the matroid associated to a $[n,k]$-linear code). As seen in Corollary~\ref{cor:greedyamongbetti}, the sets that compute the different greedy weights are to be found in the sets that have non-zero Betti numbers. Together with the main result from~\cite{Johnsen13}, this suggests that all information about various kinds of greedy weights might be encoded in minimal free resolutions of the associated Stanley-Reisner ring. This is what we will look into in the first part of this section. In the second part, we will look into the concept of chained codes and matroids.

\subsection{Greedy weights from strands}

In the rest of this section, if $M$ is a matroid on the finite set $E$ of cardinality $n$, then $S$ denotes the polynomial ring $\K[e,\ e \in E]$. This ring is naturally $\N^{n}$ and $\N$ graded.

\begin{defn}\label{def:ministrand}
Let \[f: \bigoplus_{\sigma \in \N^n} S(-\sigma)^{a_\sigma} \rightarrow \bigoplus_{\sigma \in \N^n}S(-\sigma)^{b_\sigma}\] and $\rho,\mu \in \N^n$. Then \[f_{\rho,\mu} : S(-\rho)^{a_\rho} \hookrightarrow \bigoplus_{\sigma \in \N^n} S(-\sigma)^{a_\sigma} \rightarrow \bigoplus_{\sigma \in \N^n}S(-\sigma)^{b_\sigma} \twoheadrightarrow S(-\mu)^{b_\mu}.\] Similarly, in the $\N$-graded context, let \[g: \bigoplus_{i \in \N} S(-i)^{a_i} \rightarrow \bigoplus_{i \in \N}S(-i)^{b_i}\] and $p,q \in \N$. Then \[g_{p,q} : S(-p)^{a_p} \hookrightarrow \bigoplus_{i \in \N} S(-i)^{a_i} \rightarrow \bigoplus_{i \in \N}S(-i)^{b_i} \twoheadrightarrow S(-q)^{b_q}.\] In both cases, the leftmost map is the inclusion map, while the rightmost map is the natural projection.
\end{defn}

We are now able to define the strands of a resolution.

\begin{defn}\label{def:strand}
Let $M$ be a matroid of rank $n-k$ on a finite set of cardinality $n$. If  
\[
0 \leftarrow S_M \overset{f_0}\leftarrow S \overset{f_1}\leftarrow \bigoplus_{j \in \N}S(-j)^{\beta_{1,j}} \overset{f_2}\leftarrow  \cdots \overset{f_k}\leftarrow \bigoplus_{j \in \N}S(-j)^{\beta_{k,j}} \leftarrow 0\] is a $\N$-graded resolution, and if  $\bm{h}=(h_1,\cdots,h_k) \in \N^k$, the $\bm{h}$-strand of the resolution is the sequence \[(f_{1,0,h_1},f_{2,h_1,h_2},\cdots,f_{k,h_{k-1},h_k}).\] The strand of the resolution is the $\bm{h}$-strand with $\bm{h}=(d_1\cdots,d_k)$. \\
If \[
0 \leftarrow S_M \overset{\phi_0}\leftarrow S \overset{\phi_1}\leftarrow \bigoplus_{\sigma \in \N^n}S(-j)^{\beta_{1,\sigma}} \overset{\phi_2}\leftarrow  \cdots \overset{\phi_k}\leftarrow \bigoplus_{\sigma \in \N^n}S(-j)^{\beta_{k,\sigma}} \leftarrow 0\] is a $\N^n$-graded resolution, and if  $\bm{\sigma}=(\sigma_1,\cdots,\sigma_k) \in (\N^n)^k$, the $\bm{\sigma}$-strand of the resolution is the sequence \[(\phi_{1,(0\cdots,0),\sigma_1},\phi_{2,\sigma_1,\sigma_2},\cdots,\phi_{k,\sigma_{k-1},\sigma_k}).\] 
\end{defn}

 We have already mentioned that $\beta_{i,\sigma} \neq 0 \Rightarrow \sigma \in \{0,1\}^n$. In the sequel, we will therefore identify elements of $\{0,1\}^n$ with subsets of $E=\{1,\cdots,n\}$. The main theorem of this section, will be a consequence of the following lemma.

\begin{lemma}\label{lem:nonzeromap}
Let $M$ be a matroid of rank $n-k$ on a set $E$ of cardinality $n$. Let \[0 \leftarrow S_M \overset{\phi_0}\leftarrow S \overset{\phi_1}\leftarrow \bigoplus_{\alpha \in \N^n}S(-\alpha)^{\beta_{1,\alpha}} \overset{\phi_2}\leftarrow \cdots \overset{\phi_k}\leftarrow \bigoplus_{\alpha \in \N^n}S(-\alpha)^{\beta_{k,\alpha}} \leftarrow 0\] be any minimal $\N^n$ graded resolution of its Stanley-Reisner ring. Let $\rho,\mu$ be two subsets of $E$. Then \[\phi_{l,\rho,\mu} \neq 0 \Leftrightarrow \rho \in N_{l-1},\ \mu \in N_{l} \textrm{ and } \rho \subset \mu.\]
\end{lemma}

\begin{proof} Any minimal resolution differs from the Taylor resolution (see~\cite[Section 7.1]{Herzog11} by adding trivial resolutions of the form \[0 \leftarrow \cdots \leftarrow 0 \leftarrow S(-\sigma) \overset{\psi_j}\leftarrow S(-\sigma) \leftarrow 0 \leftarrow \cdots \leftarrow 0.\] For $\rho,\mu \subset E$, it is easy to see that if $f=g \oplus h$, then $f_{\rho,\mu} = g_{\rho_\mu} \oplus h_{\rho,\mu}$. In particular, if $\rho \neq \mu$, then \[\psi_{i,\rho,\mu}=0\] for every $i$, so that \[\phi_{l,\rho,\mu} \neq 0 \Leftrightarrow \Phi_{l,\rho,\mu} \neq 0\] where $\Psi_*$ are the maps in the Taylor resolution. In any minimal free resolution, we have \[\beta_{l,X} \neq 0 \Leftrightarrow  X \in N_l,\] so we might assume that $\rho \in N_{l-1}$ and $\mu \in N_l$, otherwise $\phi_{l,\rho,\mu} = 0$. In particular, this means that $\rho \neq \mu$. 

In a first step, we prove that \[\rho \subset \mu \Leftrightarrow \exists \tau \in \C,\ \mu = \tau \cup \rho.\] One way is obvious. For the other way, let $y \in \mu \backslash \rho$. Since $\mu$ is a cycle, there exists $\tau \in \C$ with $y \in \tau \subset \mu$. Then we have \[n(\mu) \geqslant n(\rho \cup \tau) \geqslant n(\rho) + n(\tau) - n(rho \cap \tau) = l-1\] the equality coming from the fact that $n(\rho \cap \tau)=0$ since $\rho \cap \tau \subsetneq \tau$ is strictly included in a circuit and has thus nullity $0$. Since $n(\mu)=l$ and $\mu$ is minimal, we have equality \[\rho \cup \tau = \mu.\]
Now, if $\rho \cup \tau=\mu$ and $\rho  \in N_{l-1}$, by~\cite[Proposition 1]{Johnsen13}, we can write $\rho = \bigcup_{i=1}^{l-1} \sigma_i$ for some distinct circuits $\sigma_i$, and by construction of the differential of the Taylor complex, \[\Phi_{l,\rho,\mu} \neq 0.\] Conversely, if $\Psi_{l,\rho,\mu} \neq 0$, then again by construction of the differential of the Taylor complex, $\mu$ is the union of $l$ circuits, and we obtain  $\rho$ by taking the union of all these circuits but $1$. 

\qed \end{proof}

We then have:
\begin{thm} \label{nonzerocomp} 
Let $M$ be a matroid of rank $n-k$ on a set of cardinality $n$. Let \[
0 \leftarrow S_M \overset{f_0}\leftarrow S \overset{f_1}\leftarrow \bigoplus_{j \in \N}S(-j)^{\beta_{1,j}} \overset{f_2}\leftarrow  \cdots \overset{f_k}\leftarrow \bigoplus_{j \in \N}S(-j)^{\beta_{k,j}} \leftarrow 0\] and \[
0 \leftarrow S_M \overset{\phi_0}\leftarrow S \overset{\phi_1}\leftarrow \bigoplus_{\sigma \in \N^n}S(-j)^{\beta_{1,\sigma}} \overset{\phi_2}\leftarrow  \cdots \overset{\phi_k}\leftarrow \bigoplus_{\sigma \in \N^n}S(-j)^{\beta_{k,\sigma}} \leftarrow 0\] be $\N$ and $\N^n$-graded resolutions respectively. Then
\begin{enumerate}
\item $e_1=g_1=d_e= \min\{j, \ \beta_{1,j} \neq 0\}$ and $\tilde{e}_k = \min \{j,\ \beta_{k,j} \neq 0\}.$
\item For $2\leqslant l \leqslant k$, the greedy weight $e_l$ is \[e_l=\min \{|\sigma|,\ \exists \tau \textrm{ that computes }e_{l-1},\ \phi_{l,\tau,\sigma} \neq 0\}.\]
\item For $1\leqslant l \leqslant k-1$, the top down greedy weight $\tilde{e}_l$ is \[e_l=\min \{|\sigma|,\ \exists \tau \textrm{ that computes }\tilde{e}_{l+1},\ \phi_{l,\sigma,\tau} \neq 0\}.\]
\item For $2 \leqslant l \leqslant k$, the CEZ greedy weight $g_l$ is \[g_l = \min \{j,\  f_{l,d_{l-1},j} \neq 0\}.\]
\item \[(e_1,\cdots,e_k)= \min_{lex} \{e(\bm{\sigma}),\ \bm{\sigma} \in (2^E)^n, \textrm{ the }\bm{\sigma}\textrm{-strand consists of non-zero maps}\}.\]
\item \[(\tilde{e}_1,\cdots,\tilde{e}_k)= \min_{revlex} \{e(\bm{\sigma}),\ \bm{\sigma} \in (2^E)^n, \textrm{ the }\bm{\sigma}\textrm{-strand consists of non-zero maps}\}.\]
\end{enumerate}

\end{thm} 
\begin{proof}
The first point is just the definition. The second and third point are consequences of the previous lemma. The fourth point is also a consequence of the previous lemma. Here, we can take the $\N$-graded resolution, since any subset of cardinality $d_{l-1}$ with non-zero Betti number computes $d_{l-1}$. The two last points follow from the second and third point, as well as Proposition~\ref{prop:2defs}.
\qed \end{proof}

\begin{example}
Using for example \cite{Novik02}, we are able to compute the $\N^n$-graded resolution of the code of Example~\ref{exa:code}. 
\setcounter{MaxMatrixCols}{7}
\[\begin{bmatrix}{6}{7}{8} & {5}{7}{8} & {5}{6}{8} & 
{5}{6}{7} & {2}{3}{4} & {1}{3}{4} & {1}{2}
\end{bmatrix}\]

\[S \longleftarrow \begin{array}{c}S(-{6}{7}{8}) \bigoplus S(-{5}{7}{8}) \bigoplus S(-{5}{6}{8})\bigoplus S(-{5}{6}{7})\\ \bigoplus S(-{2}{3}{4}) \bigoplus S(-{1}{3}{4}) \bigoplus S(-{1}{2})\end{array} \]

\setcounter{MaxMatrixCols}{17}
\begin{small}\[\begin{bmatrix}-{2}{3}{4} & 0 & 0 & 0 & -{1}{3}{4} & 0 & 0 & 0 & 
-{1}{2} & 0 & 0 & 0 & -{5} & -{5} & -{5} & 0 & 0 \\
 0 & -{2}{3}{4} & 0 & 0 & 0 & -{1}{3}{4} & 0 & 0 & 0 & -{1}{2} &
0 & 0 & {6} & 0 & 0 & 0 & 0 \\
 0 & 0 & -{2}{3}{4} & 0 & 0 & 0 & -{1}{3}{4} & 0 & 0 & 0 & 
-{1}{2} & 0 & 0 & {7} & 0 & 0 & 0 \\
 0 & 0 & 0 & -{2}{3}{4} & 0 & 0 & 0 & -{1}{3}{4} & 0 & 0 & 0 & 
-{1}{2} & 0 & 0 & {8} & 0 & 0 \\
 {6}{7}{8} & {5}{7}{8} & {5}{6}{8} & {5}{6}{7} & 0 & 0 &
0 & 0 & 0 & 0 & 0 & 0 & 0 & 0 & 0 & -{1} & -{1} \\
 0 & 0 & 0 & 0 & {6}{7}{8} & {5}{7}{8} & {5}{6}{8} & 
{5}{6}{7} & 0 & 0 & 0 & 0 & 0 & 0 & 0 & {2} & 0 \\
 0 & 0 & 0 & 0 & 0 & 0 & 0 & 0 & {6}{7}{8} & {5}{7}{8} & 
{5}{6}{8} & {5}{6}{7} & 0 & 0 & 0 & 0 & {3}{4}
\end{bmatrix}\]\end{small}
\[\begin{array}{c}S(-{6}{7}{8}) \bigoplus S(-{5}{7}{8}) \bigoplus S(-{5}{6}{8}) \\
\bigoplus S(-{5}{6}{7}) \bigoplus S(-{2}{3}{4})\\ \bigoplus 
S(-{1}{3}{4}) \bigoplus S(-{1}{2})\end{array} \longleftarrow 
\begin{array}{c}S(-{2}{3}{4}{6}{7}{8}) \bigoplus S(-{2}{3}{4}{5}{7}{8}) 
\bigoplus S(-{2}{3}{4}{5}{6}{8})\\ \bigoplus 
S(-{2}{3}{4}{5}{6}{7}) \bigoplus S(-{1}{3}{4}{6}{7}{8}) 
\bigoplus S(-{1}{3}{4}{5}{7}{8})\\ \bigoplus 
S(-{1}{3}{4}{5}{6}{8}) \bigoplus S(-{1}{3}{4}{5}{6}{7}) 
\bigoplus S(-{1}{2}{6}{7}{8}) \\\bigoplus S(-{1}{2}{5}{7}{8}) 
\bigoplus S(-{1}{2}{5}{6}{8}) \bigoplus S(-{1}{2}{5}{6}{7}) \\
\bigoplus S(-{5}{6}{7}{8})^{3}  \bigoplus S(-{1}{2}{3}{4})^{2}\end{array}  \]

\setcounter{MaxMatrixCols}{17}
\[\begin{bmatrix}-{5} & -{5} & -{5} & 0 & 0 & 0 & {1} & {1} & 0 & 0 & 0 
& 0 & 0 & 0 & 0 & 0 & 0 \\
 {6} & 0 & 0 & 0 & 0 & 0 & 0 & 0 & {1} & {1} & 0 & 0 & 0 & 0 & 0 & 0 & 0 
\\
 0 & {7} & 0 & 0 & 0 & 0 & 0 & 0 & 0 & 0 & {1} & {1} & 0 & 0 & 0 & 0 & 0 
\\
 0 & 0 & {8} & 0 & 0 & 0 & 0 & 0 & 0 & 0 & 0 & 0 & {1} & {1} & 0 & 0 & 0 
\\
 0 & 0 & 0 & -{5} & -{5} & -{5} & -{2} & 0 & 0 & 0 & 0 & 0 & 0 & 0 & 0 &
0 & 0 \\
 0 & 0 & 0 & {6} & 0 & 0 & 0 & 0 & -{2} & 0 & 0 & 0 & 0 & 0 & 0 & 0 & 0 \\
 0 & 0 & 0 & 0 & {7} & 0 & 0 & 0 & 0 & 0 & -{2} & 0 & 0 & 0 & 0 & 0 & 0 \\
 0 & 0 & 0 & 0 & 0 & {8} & 0 & 0 & 0 & 0 & 0 & 0 & -{2} & 0 & 0 & 0 & 0 \\
 0 & 0 & 0 & 0 & 0 & 0 & 0 & -{3}{4} & 0 & 0 & 0 & 0 & 0 & 0 & -{5} & 
-{5} & -{5} \\
 0 & 0 & 0 & 0 & 0 & 0 & 0 & 0 & 0 & -{3}{4} & 0 & 0 & 0 & 0 & {6} & 0 & 0
\\
 0 & 0 & 0 & 0 & 0 & 0 & 0 & 0 & 0 & 0 & 0 & -{3}{4} & 0 & 0 & 0 & {7} & 0
\\
 0 & 0 & 0 & 0 & 0 & 0 & 0 & 0 & 0 & 0 & 0 & 0 & 0 & -{3}{4} & 0 & 0 & {8}
\\
 {2}{3}{4} & 0 & 0 & {1}{3}{4} & 0 & 0 & 0 & 0 & 0 & 0 & 0 & 0 & 0 &
0 & {1}{2} & 0 & 0 \\
 0 & {2}{3}{4} & 0 & 0 & {1}{3}{4} & 0 & 0 & 0 & 0 & 0 & 0 & 0 & 0 &
0 & 0 & {1}{2} & 0 \\
 0 & 0 & {2}{3}{4} & 0 & 0 & {1}{3}{4} & 0 & 0 & 0 & 0 & 0 & 0 & 0 &
0 & 0 & 0 & {1}{2} \\
 0 & 0 & 0 & 0 & 0 & 0 & {6}{7}{8} & 0 & {5}{7}{8} & 0 & 
{5}{6}{8} & 0 & {5}{6}{7} & 0 & 0 & 0 & 0 \\
 0 & 0 & 0 & 0 & 0 & 0 & 0 & {6}{7}{8} & 0 & {5}{7}{8} & 0 & 
{5}{6}{8} & 0 & {5}{6}{7} & 0 & 0 & 0
\end{bmatrix}\]
\[\begin{array}{c}S(-{2}{3}{4}{6}{7}{8}) \bigoplus S(-{2}{3}{4}{5}{7}{8}) 
\bigoplus S(-{2}{3}{4}{5}{6}{8}) \\\bigoplus 
S(-{2}{3}{4}{5}{6}{7}) \bigoplus S(-{1}{3}{4}{6}{7}{8}) 
\bigoplus S(-{1}{3}{4}{5}{7}{8})\\ \bigoplus 
S(-{1}{3}{4}{5}{6}{8}) \bigoplus S(-{1}{3}{4}{5}{6}{7}) 
\bigoplus S(-{1}{2}{6}{7}{8})\\ \bigoplus S(-{1}{2}{5}{7}{8}) 
\bigoplus S(-{1}{2}{5}{6}{8}) \bigoplus S(-{1}{2}{5}{6}{7}) \\
\bigoplus S(-{5}{6}{7}{8})^{3}  \bigoplus S(-{1}{2}{3}{4})^{2}  \end{array}
\longleftarrow \begin{array}{c}S(-{2}{3}{4}{5}{6}{7}{8})^{3}  \bigoplus 
S(-{1}{3}{4}{5}{6}{7}{8})^{3} \\ \bigoplus 
S(-{1}{2}{3}{4}{6}{7}{8})^{2}  \bigoplus 
S(-{1}{2}{3}{4}{5}{7}{8})^{2} \\ \bigoplus 
S(-{1}{2}{3}{4}{5}{6}{8})^{2}  \bigoplus 
S(-{1}{2}{3}{4}{5}{6}{7})^{2} \\ \bigoplus 
S(-{1}{2}{5}{6}{7}{8})^{3} \end{array} \]

\setcounter{MaxMatrixCols}{6}
\[\begin{bmatrix}-{1} & 0 & 0 & -{1} & 0 & 0 \\
 0 & -{1} & 0 & 0 & -{1} & 0 \\
 0 & 0 & -{1} & 0 & 0 & -{1} \\
 {2} & 0 & 0 & 0 & 0 & 0 \\
 0 & {2} & 0 & 0 & 0 & 0 \\
 0 & 0 & {2} & 0 & 0 & 0 \\
 -{5} & -{5} & -{5} & 0 & 0 & 0 \\
 0 & 0 & 0 & -{5} & -{5} & -{5} \\
 {6} & 0 & 0 & 0 & 0 & 0 \\
 0 & 0 & 0 & {6} & 0 & 0 \\
 0 & {7} & 0 & 0 & 0 & 0 \\
 0 & 0 & 0 & 0 & {7} & 0 \\
 0 & 0 & {8} & 0 & 0 & 0 \\
 0 & 0 & 0 & 0 & 0 & {8} \\
 0 & 0 & 0 & {3}{4} & 0 & 0 \\
 0 & 0 & 0 & 0 & {3}{4} & 0 \\
 0 & 0 & 0 & 0 & 0 & {3}{4}
\end{bmatrix}\]
\[
\begin{array}{c}S(-{2}{3}{4}{5}{6}{7}{8})^{3}  \bigoplus 
S(-{1}{3}{4}{5}{6}{7}{8})^{3}  \bigoplus 
S(-{1}{2}{3}{4}{6}{7}{8})^{2} \\ \bigoplus 
S(-{1}{2}{3}{4}{5}{7}{8})^{2}  \bigoplus 
S(-{1}{2}{3}{4}{5}{6}{8})^{2}  \bigoplus 
S(-{1}{2}{3}{4}{5}{6}{7})^{2}\\  \bigoplus 
S(-{1}{2}{5}{6}{7}{8})^{3} \end{array} \longleftarrow 
S(-{1}{2}{3}{4}{5}{6}{7}{8})^{6}  \]
where, for compactifying the notation, the number $a_1\ldots a_t$ represents the monomial $X_{a_1}\cdots X_{a_t}$. Here the $(\{1,2\},\{1,2,3,4\},\{1,2,3,4,6,7,8\},E)$-strand is\[\left(\begin{bmatrix}12\end{bmatrix}, \begin{bmatrix} 0 & 34\end{bmatrix}, \begin{bmatrix} 678 & 0 \\ 0 & 678 \end{bmatrix} , \begin{bmatrix} -5&-5&-5&0&0&0\\ 0&0&0&-5&-5&-5\end{bmatrix}\right),\] thus \[(e_1,e_2,e_3,e_4) \leqslant_{lex} (2,4,7,8)\] and there is actually equality here.
\end{example}

\subsection{Chained codes and matroids}

\begin{defn}\label{def:chainedcode}
Let $C$ be a linear code of dimension $k$. It is called chained if there is a chain \[D_1 \subset D_2 \subset \cdots\ \subset D_k\] of linear subcodes, such that $D_i$ computes $d_i$, for $1 \leqslant i \leqslant k$.
\end{defn}

Then we have: 
\begin{prop}\label{prop:chainedcode}
Let $C$ be a linear code of dimension $k$. Then the following assertions are equivalent: \begin{itemize}
\item The code $C$ is chained,
\item $(d_1,\cdots,d_k)=(e_1,\cdots,e_k)$,
\item $(d_1,\cdots,d_k)=(\tilde{e}_1,\cdots,\tilde{e}_k)$,
\item $(d_1,\cdots,d_k)=(g_1,\cdots,g_k)$.
\end{itemize}
\end{prop}

\begin{proof} This is obvious from the definitions.
\qed \end{proof}

This can be generalized to matroids:

\begin{defn}\label{def:chainedmatroid}
A matroid of rank $n-k$ on a set of cardinality $n$ is chained if there exists a chain \[\sigma_1\subset \cdots \subset \sigma_k\] where $\sigma_i \in N_i$ computes $d_i$.
\end{defn}

\begin{prop}\label{prop:chainedmatroid}
Let $M$ be a matroid of rank $n-k$ on a set of cardinality $n$. then the following assertions are equivalent:
\begin{itemize}
\item The matroid $M$ is chained,
\item $(d_1,\cdots,d_k)=(e_1,\cdots,e_k)$,
\item $(d_1,\cdots,d_k)=(\tilde{e}_1,\cdots,\tilde{e}_k)$,
\item $(d_1,\cdots,d_k)=(g_1,\cdots,g_k)$.
\end{itemize}
\end{prop}

\begin{proof}This follows from the definitions.
\qed \end{proof}
Moreover, we have the following:
\begin{prop}\label{prop:coincide}
A linear code is chained if and only if its associated matroid is chained.
\end{prop}
\begin{proof} This is a direct consequence of Theorem~\ref{th:coincide}.
 \ref{th:coincide}.
\qed \end{proof}

We will end this article with commenting on the connection between chainedness of a code or matroid, and properties of minimal resolutions of their Stanley-Reisner rings.

\begin{defn}\label{def:chained}
Let $M$ be a matroid of rank $n-k$ on a set of cardinality $n$. It has a pure resolution if the $\N$-graded resolution satisfies: \[\forall 1 \leqslant i \leqslant k,\ \exists! j_i,\ \beta_{i,j_i}\neq 0.\] Furthermore, we say that the pure resolution is linear if \[\forall 1\leqslant i <k,\ j_{i+1} = j_i.\] A linear code has pure resolution (resp. linear resolution) if its associated matroid has.
\end{defn}
\begin{prop} \label{prop:purechained}
Linear codes and matroids with pure resolutions are chained.
\end{prop}
\begin{proof} This follows from the fact that $\beta_{i,X} \neq 0 \Leftrightarrow X \in N_i$ and the definitions.
\qed \end{proof}

MDS codes and uniform matroids have linear resolutions, and as such are chained. On the other hand, we have some codes that do not have linear resolutions, but pure resolutions. Examples of such codes are Reed-M\"uller of the first order and constant weight codes (\cite[Theroem 2.1]{Johnsen14}). These codes are thus also chained. In the case of constant weight codes, we have in addition that any chain of subcodes of dimension $i$ actually compute $e_i=d_i$. In general, chained codes do not need to have pure resolutions. For example, non-binary Veronese codes studied in~\cite{Johnsen20} are such  codes. These codes are defined through the Veronese embedding $\mathbb{P}^2 \rightarrow \mathbb{P}^5$. Elements of $N_i$ correspond to complements of geometrical configurations, and it follows easily from~\cite[Theorem 21]{Johnsen20} that the code is chained but does not have pure resolution.\\
The set of codes/matroids with pure/linear resolutions is not closed under taking duals. However, we have:

\begin{prop}\label{prop:dualchained}
A matroid (resp. linear code) is chained if and only if its dual (resp. orthogonal complement) is chained.
\end{prop}

\begin{proof} This follows from Wei duality for greedy weights and Hamming weights.
\qed \end{proof}

\bibliographystyle{abbrv}
\bibliography{biblio}

\begin{thebibliography}{10}

\bibitem{Bai19}
L.~Bai and Z.~Liu.
\newblock On the third greedy weight of $4$-dimensional codes.
\newblock {\em Designs, Codes and Cryptography}, 87(10):2213--2230, 2019.

\bibitem{Chen1999}
W.~Chen and T.~Kl{\o}ve.
\newblock On the second greedy weight for binary linear codes.
\newblock In {\em International Symposium on Applied Algebra, Algebraic
  Algorithms, and Error-Correcting Codes}, pages 131--141. Springer, 1999.

\bibitem{chen1999weight}
W.~Chen and T.~Klove.
\newblock Weight hierarchies of extremal non-chain binary codes of dimension 4.
\newblock {\em IEEE Transactions on Information Theory}, 45(1):276--281, 1999.

\bibitem{Chen01}
W.~Chen and T.~Kl{\o}ve.
\newblock On the second greedy weight for linear codes of dimension 3.
\newblock {\em Discrete Mathematics}, 241(1-3):171--187, 2001.

\bibitem{Chen04}
W.~Chen and T.~Kl{\o}ve.
\newblock On the second greedy weight for linear codes of dimension at least 4.
\newblock {\em IEEE Transactions on Information Theory}, 50(2):354--356, 2004.

\bibitem{Cohen1999}
G.~D. Cohen, S.~B. Encheva, and G.~Z\'emor.
\newblock Antichain codes.
\newblock {\em Designs, Codes and Cryptography}, 18(1-3):71--80, 1999.

\bibitem{GhorpadeSingh20}
S.~Ghorpade and P.~Singh.
\newblock Pure resolutions, linear codes, and {B}etti numbers.
\newblock {\em arXiv 2002.01799}, pages 1--22, 2020.

\bibitem{Herzog11}
J.~Herzog and T.~Hibi.
\newblock {\em Monomial ideals}.
\newblock Springer-Verlag, London, 2011.
\newblock Graduate Texts in Mathematics, No. 260.

\bibitem{Johnsen13}
T.~Johnsen and H.~Verdure.
\newblock Hamming weights and {B}etti numbers of {S}tanley--{R}eisner rings
  associated to matroids.
\newblock {\em Applicable Algebra in Engineering, Communication and Computing},
  24(1):73--93, 2013.

\bibitem{Johnsen14}
T.~Johnsen and H.~Verdure.
\newblock Stanley--{R}eisner resolution of constant weight linear codes.
\newblock {\em Designs, Codes and Cryptography}, 72(2):471--481, 2014.

\bibitem{Johnsen20}
T.~Johnsen and H.~Verdure.
\newblock Higher weight spectra of {V}eronese codes.
\newblock {\em IEEE Transactions on information theory}, 2019.
\newblock \mbox{DOI}:\url{10.1109/TIT.2019.2948180}, to appear.

\bibitem{Larsen05}
A.~H. Larsen.
\newblock Matroider og line{\ae}re koder.
\newblock Master's thesis, University of Bergen, 2005.
\newblock http://bora.uib.no/handle/1956/10780.

\bibitem{Li17}
X.~Li and Z.~Liu.
\newblock On the second relative greedy weight.
\newblock {\em Cryptography and Communications}, 9(2):181--197, 2017.

\bibitem{miller04}
E.~Miller and B.~Sturmfels.
\newblock {\em Combinatorial commutative algebra}, volume 227.
\newblock Springer Science \& Business Media, 2004.

\bibitem{Novik02}
I.~Novik, A.~Postnikov, and B.~Sturmfels.
\newblock Syzygies of oriented matroids.
\newblock {\em Duke Mathematical Journal}, 111(2):287--317, 2002.

\bibitem{Oxley11}
J.~G. Oxley.
\newblock {\em Matroid theory, Second edition}.
\newblock Oxford University Press, Oxford, 2011.
\newblock Oxford Graduate Texts in Mathematics, 21.

\bibitem{Schaathun01}
H.~G. Schaathun.
\newblock Duality and greedy weights of linear codes and projective multisets.
\newblock In {\em International Symposium on Applied Algebra, Algebraic
  Algorithms, and Error-Correcting Codes}, pages 92--101. Springer, 2001.

\bibitem{Schaathun01phd}
H.~G. Schaathun.
\newblock {\em Support weights in Linear Codes and Projective Multisets}.
\newblock PhD thesis, University of Bergen, 2001.

\bibitem{Schaathun04}
H.~G. Schaathun.
\newblock A lower bound on the greedy weights of product codes.
\newblock {\em Designs, Codes and Cryptography}, 31(1):27--42, 2004.

\bibitem{stanley77}
R.~P. Stanley.
\newblock Cohen-{M}acaulay complexes.
\newblock In {\em Higher Combinatorics}, volume pp. 51-63. Cambridge University
  Press, 2015.

\bibitem{Wei1991}
V.~K. Wei.
\newblock Generalized {H}amming weights for linear codes.
\newblock {\em IEEE Transactions on information theory}, 37(5):1412--1418,
  1991.

\end{thebibliography}

\end{document}